\documentclass[11pt,a4paper]{article}

\usepackage[margin=2.0cm]{geometry}
\usepackage[utf8]{inputenc}
\usepackage[english]{babel}

\usepackage{amsmath,amssymb,amsthm,amstext,mathtools,authblk}
\usepackage{bm}
\usepackage{mathrsfs}
\usepackage{latexsym}
\usepackage{physics}

\usepackage{graphicx}
\usepackage{placeins}
\usepackage[usenames,dvipsnames]{xcolor}

\usepackage[normalem]{ulem}
\usepackage{soul}
\setstcolor{red}

\usepackage[colorlinks=true,citecolor=blue,urlcolor=blue,linkcolor=blue]{hyperref}

\DeclarePairedDelimiter{\ceil}{\lceil}{\rceil}

\newcommand{\R}{\mathbb{R}}
\DeclareMathOperator{\vol}{vol}

\newtheorem{theorem}{Theorem}
\newtheorem{lemma}[theorem]{Lemma}
\newtheorem{proposition}[theorem]{Proposition}

\theoremstyle{definition}

\theoremstyle{remark}

\newcommand{\rmd}{\mathrm{d}}

\begin{document}

\title{\bf Analysis of Conjectural Improvements to Minkowski’s\\ Lower Bound on the Sphere Packing Density}

\author[1]{Carlo Vanoni}
\author[2]{Alan Guo}
\author[1,3,4,5]{Salvatore Torquato \thanks{torquato@princeton.edu}}

\affil[1]{Department of Physics, Princeton University, Princeton, New Jersey, 08544, USA}
\affil[2]{Independent Researcher, Newton, Massachusetts, USA}
\affil[3]{Department of Chemistry, Princeton University, Princeton, New Jersey, 08544, USA}
\affil[4]{Princeton Institute for the Science and Technology of Materials, Princeton University, Princeton, New Jersey 08544, USA}
\affil[5]{Program in Applied and Computational Mathematics, Princeton University, Princeton, New Jersey 08544, USA}

\date{\today}

\maketitle

\begin{abstract}
Torquato and Stillinger conjectured an exponential improvement of Minkowski’s classical lower bound on the maximal density of sphere packings in high-dimensional Euclidean space $\mathbb{R}^d$, with exponential rate $2^{-(0.7786524795\ldots+o(1))d}$, using a pair-correlation-function optimization framework. Conditional on their realizability conjecture, we show that a simple family of hyperuniform pair correlation functions yields polynomial improvements over Minkowski’s lower bound of the form $\phi_{\mathrm{max}}\gtrsim d^\beta 2^{-d}$ for every fixed $\beta>1$ in the high-$d$ limit. As the polynomial exponent is allowed to increase with dimension, this family continuously approaches the previously conjectured exponential improvement.
We next derive the Torquato--Stillinger exponential rate independently from the Cohn--Elkies dual linear-programming upper-bound formulation, showing that its radial test functions cannot asymptotically exclude packings with the Torquato--Stillinger density scaling.
We then examine the robustness of this rate across a broad range of radial pair-correlation test functions.
For the near-contact families considered, modifying a fixed number of shells, gaps, or radial bands can improve subexponential factors but does not improve the leading exponential rate.
To surpass this rate, we introduce an explicit hyperuniform construction based on Gauss--Radau quadrature.
It contains $\lfloor(d-1)/4\rfloor$ positive delta-function shells and has exponential rate $2^{-(0.622556248918\ldots+o(1))d}$, providing a concrete mechanism by which growing radial complexity surpasses the Torquato--Stillinger rate.
Finally, we prove strong duality between the unrestricted pair-correlation program and its Cohn--Elkies dual: their optimal values coincide, with no duality gap.
Combining this result with approximation by ordinary finite-band functions, we show that the unrestricted pair-correlation program attains the optimal Cohn--Elkies exponential rate $2^{-(0.6044005\ldots+o(1))d}$ in the high-$d$ limit.
This result is general and optimal but does not supply a comparable closed-form family.
All of these various results taken together provide new evidence that exceptionally dense disordered sphere packings may exist in high dimensions and strengthen the case for the general Torquato--Stillinger conjecture in the high-$d$ limit.

\end{abstract}

\section{Introduction}

Sphere-packing problems remain a vibrant area of research, with far-reaching applications throughout the physical, mathematical, and biological sciences~\cite{To26b,To18b,To10c,Du19,Ro24,Chaik95}. In physics, disordered sphere packings in dimensions $d$ greater than three provide valuable models for understanding liquids, metastable states, and glasses~\cite{Pa10,Ch17,Pa06,Ku12,Cha14,Ch22}, while also offering insight into the behavior of disordered phases in lower-dimensional systems. More broadly, the large-dimensional limit has emerged as a powerful theoretical framework, much as it has in other areas of physics, including the study of certain aspects of quantum gravity~\cite{Af20,Ha19}.
The quest to determine the densest sphere packings in high-dimensional Euclidean spaces is of fundamental interest to both mathematicians and physicists. 
The mathematical study of the densest sphere packings across spatial dimensions has long been a central theme in discrete geometry~\cite{Co93,Pa10}, with deep connections to number theory~\cite{Co93,Grah03}, cryptography~\cite{Co93}, and coding theory~\cite{Sh48,Co93}. Beyond its intrinsic mathematical significance, this problem has profound practical implications: the optimal transmission of digital information over a noisy communication channel is mathematically equivalent to finding the densest sphere packing in a sufficiently high-dimensional Euclidean space~\cite{Sh48,Sh49,Co93}. The resulting error-correcting codes underpin modern digital communication and data storage technologies, including compact disks, cellular networks, and the Internet.

Our focus in this paper is the sphere packing problem, which asks for the densest packing of identical nonoverlapping spheres in $d$-dimensional Euclidean space $\mathbb{R}^d$  at number density $\rho$ (number of sphere centers per unit volume), i.e., what fraction of  $\mathbb{R}^d$ can be covered by balls that do not intersect except along their boundaries? 
We call this the packing fraction $\phi = \rho \vol(B_{1/2}^d)$, where $\vol(B_{1/2}^d)=\pi^{d/2}/(2^d \Gamma(1+d/2))$ is the volume of the unit-diameter ball~\cite{To02c}.

Torquato and Stillinger formulated a linear program (LP) to obtain lower bounds on the maximal packing fraction, which we denote by  $\phi_{\mbox{\scriptsize max}}$ 
\cite{To02c}. This formulation was subsequently applied to support the counterintuitive conjecture that the densest sphere packings for sufficiently large $d$ may be disordered or at least possess fundamental cells whose size and structural complexity increase with dimension~\cite{To06b}.
Specifically, for a fixed radial test function $g_2(r)$, one determines the largest packing fraction $\phi_*$ such that $g_2(r)\geq 0$ and the associated structure factor satisfies $S(k)\geq 0$ for all $k$, where $S(k) =1+\rho {\tilde h}(k)$ and ${\tilde h}(k)$ is the Fourier transform of the total correlation function $h(r)\equiv g_2(r)-1$. If this $g_2$ is realizable by sphere packings throughout the interval $0<\phi\leq\phi_*$, then one obtains the lower bound
\begin{equation}
\phi_{\mbox{\scriptsize max}} \ge \phi_*\;,
\label{lower-bound}
\end{equation}
where $\phi_*$ is called the {\it terminal} packing fraction.
This optimization formulation is the {\it dual} of the infinite-dimensional LP devised by Cohn and Elkies~\cite{Co03} to obtain upper bounds on $\phi_{\mbox{\scriptsize max}}$. 
This duality was identified by Torquato and Stillinger~\cite{To06b}; the absence of a duality gap in the generalized measure formulation was subsequently proved by Cohn, de Laat, and Salmon~\cite[Theorem~3.1]{Co22}. The hard-core condition and the nonnegativity of $g_2(r)$ and $S(k)$ are precisely the constraints of this dual program, as described in Sec.~\ref{sec:finite-band-construction}.
Here, for spheres of unit diameter, one considers a radial test function $f(r)$ such that $f(r)\leq 0$ for $r\geq 1$ and ${\tilde f}(k)\geq 0$ for all $k$,
where the radial function ${\tilde f}(k)$ is the Fourier transform of $f(r)$. The maximal number density $\rho_{\mbox{\scriptsize max}}$ is then bounded from above  by $f(0)/{\tilde f}(0)$, or equivalently, the maximal packing fraction satisfies
\begin{equation}
    \phi_{\mbox{\scriptsize max}} \leq \vol(B_{1/2}^d)\frac{f(0)}{{\tilde f}(0)}
    = \frac{\pi^{d/2} f(0)}{2^d\Gamma(1+d/2){\tilde f}(0)}.
\end{equation}
The Cohn-Elkies LP formulation was used to prove that
the $E_8$ and Leech lattices are the densest packings in $\mathbb{R}^8$
and $\mathbb{R}^{24}$, respectively \cite{Vi17,Co17}.

Torquato and Stillinger~\cite{To06b} conjectured that a test radial function $g_2(r)$ is a realizable pair correlation function of a translationally invariant disordered sphere packing in $\mathbb{R}^d$ for $0 \le \phi \le \phi_*$ in the high-$d$ limit if and only if the nonnegativity conditions on $S(k)$ and $g_2(r)$ are met.
In other words, while there exist other necessary conditions for the realizability of disordered point configurations in sufficiently low dimensions, the nonnegativity of $g_2(r)$ and $S(k)$ is conjectured to become sufficient in the high-$d$ limit, i.e., as is the case for Gaussian random fields in any dimension. 
This conjecture is based on the so-called {\it decorrelation} principle \cite{To06b,An16}, which states that unconstrained correlations in disordered sphere packings vanish asymptotically in high dimensions, and the $n$-particle correlation functions, $g_n$ for any $n \ge 3$, can be inferred entirely (up to small errors) from a knowledge of $\rho$ and the pair correlation function $g_2$, which we know cannot be true in sufficiently low dimensions~\cite{Kuna11}.
If the decorrelation principle is true, then the inequality (\ref{lower-bound}) is a rigorous lower bound. 
It is noteworthy that the decorrelation principle is exhibited for a variety of different disordered packing models \cite{To26b}, even as the dimension increases for relatively small $d$. 
For example, for maximally random jammed (MRJ) states, spatial correlations have been shown to diminish as the space dimension
increases from three to six \cite{Sk06}. This low-$d$
behavior in MRJ packings is consistent with a high-dimensional asymptotic scaling of 
a pair correlation function that becomes a contact delta function 
plus a flat background beyond the hard core, yielding a lower bound of the form $d/2^d$, as described immediately below.

As we now discuss, there is strong evidence to support the validity of this conjecture.
First, we note that other necessary conditions [beyond the nonnegativity of $g_2(r)$ and $S(k)$], including the Yamada condition \cite{Ya61}, appear to have relevance only in low dimensions \cite{Cos04,To06b,To26b}.
Using a step function $\Theta(r-D)$ as a test pair correlation function, corresponding to a disordered packing with no correlations beyond the hard core, it is easily shown that $\phi_{\mbox{\scriptsize max}} \ge 1/2^d$ (see Refs. \cite{To02c,To06b}), which is exactly achieved by the ghost random sequential addition (ghost RSA) packing \cite{To06a} as $d\to \infty$ and is the same scaling as Minkowski's lower bound for a lattice. 
This is a vivid example that the aforementioned conjecture is true, since the test function is realizable by disordered sphere packings for all $\phi$ up to and including $1/2^d$. 
It is noteworthy that the exactly solvable ghost RSA model expresses $g_2$ and all higher-order correlation functions in terms of simple geometric quantities: the union volumes of overlapping exclusion spheres \cite{To06a}.
At its terminal packing fraction of $1/2^d$, as $d \to \infty$, each fixed-order correlation satisfies
\[
g_n(\mathbf{r}_1,\ldots,\mathbf{r}_n)
\sim \prod_{i<j} g_2(r_{ij}),
\]
so higher-order correlations contain no independent information beyond the pair statistics, up to vanishing corrections.
This provides concrete support for the decorrelation principle motivating the Torquato--Stillinger conjecture that hard-core exclusion, $g_2 \ge 0$, and $S(k) \ge 0$ suffice for realizability in sufficiently high dimensions. While this does not prove the general sufficiency claim, it is a definitive example that supports the conjecture.

An improved conjectural lower bound $(d+2)/2^{d+1}$  was obtained using a test function that again would correspond to a disordered packing that is the sum of a step function and a Dirac-delta function at contact of arbitrary weight (see Fig.~\ref{fig:g2r}), and then optimizing~\cite{To02c,To06b}.
A striking result is that this test function was shown to be numerically realizable by packings in very low dimensions, namely, $d=2$~\cite{Uc06a}, implying (because of decorrelation) that it would be achievable in any higher dimension.
Moreover, the high-$d$ asymptotic form $d/2^d$ is the same as the rigorous lower bound on lattice packings obtained by  Ball~\cite{Ball92}. 
More importantly, it has recently been proved that this asymptotic form is a rigorous lower bound on disordered packings  \cite{Je19}, providing definitive evidence that this second test function is indeed realizable by disordered packings for all $\phi$ in the interval $0<\phi\leq (d+2)/2^{d+1}$, bolstering the conjecture even further.

Finally, employing a test radial function for $g_2(r)$ corresponding to a putative disordered sphere packing (a step plus a delta function with a gap, described in the next section and illustrated in Fig.~\ref{fig:g2r}), Torquato and Stillinger found the following conjectural lower bound:
\begin{equation}
\label{eq:lower_bound}
    \phi_{\mbox{\scriptsize max}} \geq \phi^* \sim \frac{d^{1/6} e^{a_1 (d/2)^{1/3}}}{2^{2/3}D_1 \sqrt{\pi} 2^{[3-\log_2(e)]d/2}} \approx \frac{3.276100896 \, d^{1/6}\, e^{1.47292\, d^{1/3}}}{2^{0.7786524795\, d}},
\end{equation}
where $a_1 \approx 1.8557571$ and $D_1 \approx 0.108487857$,
providing the putative exponential improvement of Minkowski's lower bound.
This exponential decay rate has also been recently matched using a completely different approach (particular uncertainty principles of Fourier analysis) \cite{Ed25}.
We will describe an alternative procedure to derive the same exponential rate.
This shows that the Cohn-Elkies LP objective cannot be pushed below the Torquato-Stillinger exponential rate by this result.

Very recently, Klartag~\cite{Kl25} proved that there are lattice packings whose packing fraction is at least $c d^2/2^d$, thus providing the first improvement on Ball's bound by a factor of $d$. 
Can one devise test $g_2$'s corresponding to disordered packings that provide such quadratic improvement on Minkowski's lower bound? 
Indeed, the fact that  Minkowski's bound can be conjecturally improved exponentially implies that one should be able to devise test functions that give an improvement of the form $d^\beta$ for any $\beta >1$. In Section~\ref{sec:poly_improvement} we will show that this polynomial improvement can indeed be achieved.

In Sec.~\ref{sec:proof}, we derive the Torquato--Stillinger exponential rate from a complementary perspective by establishing a lower bound on the Cohn--Elkies upper-bound functional. 
The argument has a physical interpretation in terms of competing Fourier modes: low-wavenumber modes cannot by themselves enforce the required sign condition at the hard-core boundary, while the high-wavenumber modes needed for this cancellation are strongly suppressed in high dimensions.
This competition prevents the Cohn--Elkies upper bound from falling below the Torquato--Stillinger scale at leading exponential order. Thus, independently of the realizability conjecture, this upper-bound method cannot exclude packings with the conjectured density scaling.

In Sec.~\ref{sec:universality}, we investigate the robustness of this exponential rate under changes to the radial pair-correlation test functions.
For the families considered, modifying near-contact gaps, peaks, or a fixed number of radial shells or bands can improve finite-dimensional packing fractions and subexponential factors, but does not improve the leading exponent $0.7786524795\ldots$.
The common mechanism is that the first restrictive minimum of the structure factor remains near the same Bessel turning point, which fixes the leading exponential behavior. These results motivate the introduction of pair-correlation functions whose number of independently adjustable radial features grows with dimension, providing a route beyond the Torquato--Stillinger scaling.

We also introduce in Sec.~\ref{sec:radau-construction} an explicit family of pair correlations based on Gauss--Radau quadrature whose number of positive radial shells grows linearly with dimension.  These pair correlations satisfy the ordinary decorrelation property $g_2(r)=1$ beyond a finite radius and have nonnegative, hyperuniform structure factors. Their terminal packing fractions have the asymptotic form
\begin{equation}
    \phi_{*,d}=2^{-(0.622556248918\ldots+o(1))d},
\end{equation}
which improves on the Torquato--Stillinger exponent $0.7786524795\ldots$ while remaining above the unrestricted generalized two-point optimum $0.6044005\ldots$.

If we do not restrict to a particular form of $g_2$ other than satisfying the necessary conditions, we can achieve a better exponential rate even while still satisfying the standard decorrelation principle. It follows from strong duality applied to the Cohn-Elkies LP that the optimal exponent is $\alpha_*=0.6044005\ldots$~\cite{Co22,openai2026math}. We describe in Sec.~\ref{sec:finite-band-construction} a family of pair correlation functions that achieves this exponent. Each $g_2$ function consists of finitely many nonnegative radial bands and equals unity beyond a finite radius. Thus, attaining the optimal exponential rate is compatible with the standard decorrelation principle used here. Moreover, the radius beyond which $g_2(r)=1$ can be prescribed to grow arbitrarily slowly with dimension, provided it tends to infinity.

Section~\ref{sec:radau-construction} supplies an explicit intermediate construction: its linearly growing number of Gauss–Radau shells improves the Torquato–Stillinger exponent $0.7786524795 \ldots $ to $0.622556248918\ldots$, while preserving hyperuniformity and an exactly flat tail. Section~\ref{sec:finite-band-construction} then removes the special Gauss–Radau form. Using strong duality, cutoff and smoothing, and approximation by finitely many nonnegative radial bands, it shows that ordinary pair-correlation functions can approach the unrestricted optimum $0.6044005\ldots$. Thus Sec.~\ref{sec:radau-construction} gives an explicit mechanism, whereas Sec.~\ref{sec:finite-band-construction} establishes the optimal value for a much broader class. Moreover, in the Gauss–Radau family, the decorrelation radius remains bounded and approaches $2$ as $d\to\infty$. For the optimal-rate finite-band construction, the radius is finite in every dimension and may be chosen to diverge arbitrarily slowly with $d$.

Section~\ref{sec:poly_improvement} describes the aforementioned polynomial improvements. Section~\ref{sec:proof} derives the Torquato--Stillinger exponential rate as a lower bound on the Cohn--Elkies objective. Section~\ref{sec:universality} explains why modifying a fixed number of near-contact shells, gaps, or radial bands in the test pair-correlation functions considered here does not improve the leading Torquato--Stillinger exponential rate.
Section~\ref{sec:radau-construction} then introduces an explicit family of finite-range pair correlations that improves on the Torquato--Stillinger exponential rate. Section~\ref{sec:finite-band-construction} describes the finite-band construction at the optimal Cohn--Elkies exponent and its decorrelation properties. We make concluding remarks in Sec.~\ref{sec:discuss}.

\begin{figure}
    \centering
    \includegraphics[width=0.49\linewidth]{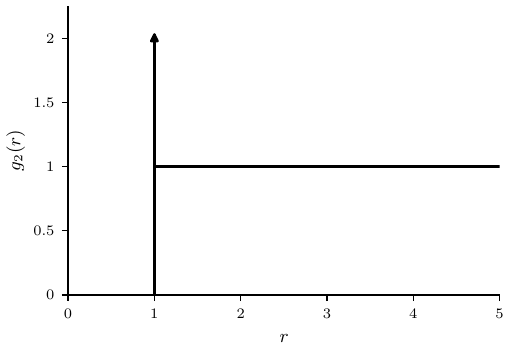}
    \includegraphics[width=0.49\linewidth]{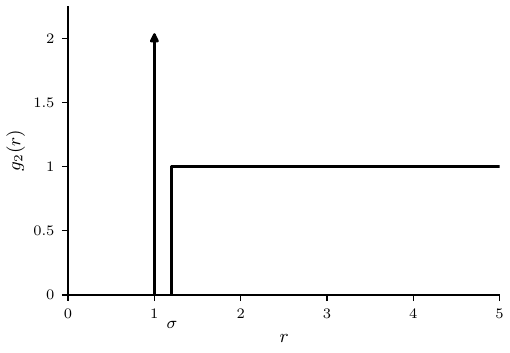}
    \caption{Radial distribution functions $g_2(r)$ constituted of a delta function at $r=1$ and a step function without a gap with the delta function (left) and with a gap $\sigma-1$ (see Eq.~\eqref{eq:g2r}) (right). The distribution function on the left provides the conjectural lower bound $\phi_{\mathrm{max}} \geq (d+2)/2^{d+1}$, whereas the one with the gap allows for obtaining the conjectural exponential improvement and the polynomial improvement $\phi_{\mathrm{max}} \gtrsim d^{\beta}/2^d$.}
    \label{fig:g2r}
\end{figure}

\section{Conjectural Polynomial Improvement on Minkowski's Lower Bound}
\label{sec:poly_improvement}

Conditional on the Torquato--Stillinger realizability conjecture, we derive a polynomially improved lower bound on the maximal packing fraction of the form $\phi_{\mathrm{max}} \gtrsim d^{\beta}/2^d$, for any $\beta > 1$ (where $\gtrsim$ hides constants independent of $d$).
We also show how this analysis can lead to the same exponential rate found by Torquato and Stillinger~\cite{To06b}.

As a first step, we introduce some basic notation that we will use throughout the paper. Following Cohn and Elkies \cite{Co03}, we  define the Fourier transform of an $L^1$ function $f \colon \R^d \to \R$ by
\begin{equation}
\tilde{f}(\xi)=\int_{\mathbb R^d} f(x)e^{-2\pi i\langle x,\xi\rangle}\,\rmd x.
\end{equation}
Note that this definition of the Fourier
transform differs from the one used by Torquato and Stillinger \cite{To06b} by a factor of $2\pi$ in the argument of the exponential function.
If $f$ is a radial function, we will write $f(r)$ for $r \ge 0$ to denote the common value $f(x)$ for any $d$-dimensional vector $x$ with $r =|x|$.
We write
\begin{equation}
    v_d=\operatorname{vol}(B_1^d)=\frac{\pi^{d/2}}{\Gamma(1+d/2)}
\end{equation}
for the volume of the unit ball. Let $d\nu_d(r)$ be radial Lebesgue measure on $\mathbb{R}^d$, i.e.,
\begin{equation}
    d\nu_d(r)=d v_d r^{d-1}\,dr,
\end{equation}
so that if $f \colon \R^d \to \R$ is radial, then
\begin{equation}
\int_{\R^d} f(x) \, \rmd x = \int_0^\infty f(r) \, \rmd\nu_d(r).
\end{equation}

Let $\sigma_d$ be the uniform probability measure on the unit sphere $S^{d-1}$ in $\R^d$, and define $B_d \colon \R^d \to \R$ by
\begin{equation}
B_d(x)
=
\int_{S^{d-1}}e^{2\pi i \langle x,y \rangle}\,\rmd\sigma_d(y).
\end{equation}
In terms of Bessel functions,
\begin{equation}
B_d(r) = \Gamma(d/2) \frac{J_{d/2-1}(2\pi r)}{(\pi r)^{d/2-1}}
\end{equation}
for $r>0$, where $J_{d/2-1}$ is the Bessel function of order $d/2-1$, and $B_d(0)=1$.
Then the Fourier transform and Fourier inversion for radial functions become
\begin{equation}
\tilde{f}(q) = \int_0^\infty f(r) B_d(rq) \, \rmd\nu_d(r)
\end{equation}
and
\begin{equation}
f(r)
=
\int_0^\infty \tilde{f}(q) B_d(rq) \, \rmd\nu_d(q).
\label{f-driect}
\end{equation}

Following Torquato and Stillinger~\cite{To06b}, we consider a test radial pair correlation function $g_2(r)$ that is a sum of a step function and a delta function with a gap (right panel of Fig.~\ref{fig:g2r}):
\begin{equation}\label{eq:g2r}
    g_2(r) = \Theta(r-\sigma) + \frac{Z}{\omega_d \rho} \delta(r-1),
\end{equation}
as shown in Fig.~\ref{fig:g2r}.
In the above expression, $Z$ denotes the average kissing number and $\omega_d$ is the surface area of the unit sphere $S^{d-1} \subset \R^d$.
The corresponding structure factor depends on the Fourier variable $q$ through the angular wavenumber
\[
    \kappa = 2\pi q.
\]
In the rest of this section, we write $S(\kappa)$ for this structure factor as a function of $\kappa$, i.e., for $S(q)$ evaluated at $q=\kappa/(2\pi)$. It is given by
\begin{equation}
    S(\kappa) = 1 + \rho \tilde{h}\left(\frac{\kappa}{2\pi} \right),
\end{equation}
where $\tilde{h}(\kappa)$ is the Fourier transform of $h(r) = g_2(r) - 1$ and $\rho$ is the number of sphere centers per unit volume. Using $g_2(r)$ given by Eq.~\eqref{eq:g2r} one finds
\begin{equation}
\label{eq:struc_fact}
    S(\kappa) = 1 - c_1(d) \frac{J_{d/2}(\kappa \sigma)}{(\kappa\sigma)^{d/2}} + c_2(d) \frac{J_{d/2-1}(\kappa)}{\kappa^{d/2-1}}.
\end{equation}
Here, the coefficients $c_1$ and $c_2$ are defined as
\begin{equation}
\label{eq:c1}
    c_1(d) = \phi \sigma^d 2^{3d/2} \Gamma(1+d/2)
\end{equation}
and
\begin{equation}
    c_2(d) = \frac{ 2^{d/2} Z\,\Gamma(1+d/2)}{d}.
\end{equation}
The putative exponential improvement obtained in Ref.~\cite{To06b} is obtained by determining the optimal values of the parameters $Z$ and $\sigma$ that maximize the packing fraction $\phi$ under the constraint that the corresponding structure factor is always nonnegative, and is also achieved by adding a second delta function~\cite{Sc08}.

\subsection{New Optimization Procedure}

Here, we will adopt a slightly different strategy in order to get the polynomial improvement using Eq.~\eqref{eq:struc_fact}. We require that $Z$ and $\sigma$ have the following dependence on $d$:
\begin{equation}
    Z = d^{\beta},
\end{equation}
\begin{equation}
    \sigma = 1+\frac{c}{d^\alpha}.
\end{equation}
Then, we optimize over $\alpha$, $\beta$, and the constant $c$ such that the constraint $S(\kappa) \geq 0$ for all $\kappa$ is always obeyed.
As a first step, we consider the expansion of Eq.~\eqref{eq:struc_fact} for small $\kappa$, which generically reads
\begin{equation}
    S(\kappa) = A_0(\sigma, Z, \phi) + A_1(\sigma, Z, \phi)\kappa^2 + \mathcal{O}(\kappa^4)
\end{equation}
where
\begin{equation}
    A_0(\sigma, Z, \phi) = 1+Z-(2\sigma)^d \phi
\end{equation}
and
\begin{equation}
    A_1(\sigma, Z, \phi) = \frac{2^{d-1}\sigma^{d+2}\phi}{d+2}-\frac{Z}{2d}.
\end{equation}
Following Ref.~\cite{To06b}, we also require that the system is hyperuniform~\cite{To03a}, i.e., $S(0)=0$ and hence $A_0 = 0$, which leads to
\begin{equation}
    \phi = \frac{\left(d^{\beta }+1\right) }{2^{d}} \left(c d^{-\alpha }+1\right)^{-d}.
\end{equation}
Plugging this into the expression for $A_1$, we get
\begin{equation}
    A_1 = \frac{1}{2} \left(\frac{\left(d^{\beta }+1\right) \left(c d^{-\alpha }+1\right)^2}{d+2}-d^{\beta -1}\right).
\end{equation}
Expanding for large $d$, we get
\begin{equation}
\label{eq:a_1}
    A_1 \sim \frac{1}{2} d^{-2 (\alpha +1)} \left(c^2 (d-2) \left(d^{\beta }+1\right)+2 c (d-2) d^{\alpha } \left(d^{\beta }+1\right)+d^{2 \alpha } \left(d-2 d^{\beta
   }\right)\right)
\end{equation}
The coefficient $A_1$ needs to be positive, which can be achieved by having a positive term diverging more quickly than $d^{2\alpha+\beta}$ in Eq.~\eqref{eq:a_1}, and this can be obtained either by $\beta<1$ for any $\alpha>0$ or $\beta > 1$ with $\alpha<1$. 
Consequently, $\alpha > 1$ leads to $\left(c d^{-\alpha }+1\right)^{-d} \to 1$ as $d\to \infty$ with $\beta<1$ to have a positive $A_1$, which is not desired for $\phi \sim d^{\beta}/2^d$. 
If instead $\alpha<1$ is considered, $\beta>1$ guarantees $A_1>0$, but $\left(c d^{-\alpha }+1\right)^{-d} \to 0$ as a stretched-exponential as $d \to \infty$, which again is not desired for $\phi$. The only choice is to set $\alpha = 1$. In this case, $\left(c d^{-1 }+1\right)^{-d} \to e^{-c}$ as $d \to \infty$, so that $\phi \sim e^{-c} d^{\beta}/2^d$, and 
\begin{equation}
    A_1 = \left(\frac{1}{2 d}+\frac{c-1}{d^2}+O\left(\left(\frac{1}{d}\right)^3\right)\right)+d^{\beta }
   \left(\frac{c-1}{d^2}+O\left(\left(\frac{1}{d}\right)^3\right)\right)
\end{equation}
is positive if $c>1$.
Using the ansatzes for $Z$ and $\sigma$ in the previous section, and setting $\alpha = 1$, we get
\begin{equation}\label{eq:Zsigmaphi}
    Z = d^{\beta}, \quad \sigma = 1 + \frac{c}{d}, \quad \phi = \frac{\left(d^{\beta }+1\right) }{2^{d}} \left(c d^{-1 }+1\right)^{-d}.
\end{equation}

We now need to ensure that $S(\kappa) \geq 0$ for all $\kappa$, which can be satisfied or not depending on the values of $\beta$, $c$, and $d$.
We start by using Eq.~\eqref{eq:Zsigmaphi} to rewrite the structure factor specified by Eq.~\eqref{eq:struc_fact}  as
\begin{equation}
    S(\kappa) = 1-\left(d^{\beta }+1\right) \, _0F_1\left(;\frac{d+2}{2};-\frac{(c+d)^2 \kappa^2}{4 d^2}\right)+d^{\beta } \, _0F_1\left(;\frac{d}{2};-\frac{\kappa^2}{4}\right),
\end{equation}
which by construction is hyperuniform, i.e.,
\begin{equation}
    \lim_{\kappa \to 0} S(\kappa) = d^{\beta } \left(\frac{2 \Gamma \left(\frac{d}{2}+1\right)}{d \, \Gamma \left(\frac{d}{2}\right)}-1\right) = 0
\end{equation}
such that $S(\kappa) \sim \kappa^2$ in the limit $\kappa\to 0$.
Therefore, the structure factor in Eq.~\eqref{eq:struc_fact} has first minimum at $\kappa=0$ such that $S(0) = 0$.

Following Ref.~\cite{To06b}, we will now impose that the other local minima of $S(\kappa)$ are nonnegative. This condition will allow us to find the values of $\beta$ for which the structure factor is globally nonnegative. 
In particular, we will show that there exists a $\beta_c(d)$ such that all $\beta < \beta_c(d)$ are valid values for dimensions greater than or equal to $d$. We find $\beta_c(d)$ by imposing that $S(\kappa_{\mathrm{min}}) = 0$, where $\kappa_{\mathrm{min}}$ is the location of the first positive minimum of the structure factor, i.e., $S'(\kappa_{\mathrm{min}}) = 0$.

Setting $S'(\kappa) = 0$ gives the equation for $\kappa_{\mathrm{min}}$
\begin{equation}
\label{eq:deriv_condition}
    \frac{J_{d/2 +1}(\kappa_{\mathrm{min}} \sigma)}{\sigma^{d/2-1}} = \frac{\kappa_{\mathrm{min}}}{d(1+d^{-\beta})}J_{d/2}(\kappa_{\mathrm{min}}),
\end{equation}
while the condition $S(\kappa_{\mathrm{min}})=0$ gives
\begin{equation}
\label{eq:Sk_condition}
    \frac{c_1(d)}{\kappa_{\mathrm{min}}^{d/2}}\left[ \frac{J_{d/2}(\kappa_{\mathrm{min}} \sigma)}{\sigma^{d/2}} - \frac{\kappa_{\mathrm{min}}}{d(1+d^{-\beta})}J_{d/2-1}(\kappa_{\mathrm{min}}) \right] = 1
\end{equation}
From Eq.~(5-26) of Ref.~\cite{To06b}, we have that the solution to Eq.~\eqref{eq:deriv_condition} can be approximated as
\begin{equation}
\label{eq:kmin}
    \kappa_{\mathrm{min}} \approx x_0 - \frac{d(1+d^{-\beta})(y_0 - \sigma x_0)}{\frac{\beta_1}{\beta_2}\sigma^{d/2 - 1} x_0 - d \sigma} \approx x_0 - \frac{d(y_0 - \sigma x_0)}{\frac{\beta_1}{\beta_2}\sigma^{d/2 - 1} x_0 - d \sigma}.
\end{equation}
Using the expansions
\begin{equation}
    x_0 \simeq d/2 + a_1 (d/2)^{1/3} + \frac{a_2}{(d/2)^{1/3}}
\end{equation}
\begin{equation}
    y_0 \simeq d/2 + a_1 (d/2)^{1/3} + 1 + \frac{a_2}{(d/2)^{1/3}} + \frac{a_2}{3(d/2)^{2/3}} + \frac{a_3}{d/2}
\end{equation}
\begin{equation}
    \frac{\beta_1}{\beta_2} \simeq 1+ \frac{2}{3d/2} - \frac{2C_2}{3C_1 (d/2)^{5/3}}
\end{equation}
where $a_1 \simeq 1.8557571$, $a_2 \simeq 1.033150$, $a_3 \simeq -0.003971$, $C_1 \simeq -1.104938082$, and $C_2 \simeq 1.627074727$, we can determine $\kappa_{\mathrm{min}} = \kappa_{\mathrm{min}}(d,\beta)$ from Eq.~\eqref{eq:kmin}. Therefore, using Eq.~\eqref{eq:Sk_condition} and Eq.~\eqref{eq:c1}, the solution $\beta_c = \beta_c(d)$ of the equation 
\begin{equation}
\label{eq:beta_c}
     2^{d/2} (1+d^{\beta_c})\Gamma(1+d/2) = \kappa_{\mathrm{min}}(d)^{d/2}\left[ \frac{J_{d/2}(\kappa_{\mathrm{min}}(d) \sigma)}{\sigma^{d/2}} - \frac{\kappa_{\mathrm{min}}(d)}{d(1+d^{-\beta_c})}J_{d/2-1}(\kappa_{\mathrm{min}}(d)) \right]^{-1}
\end{equation}
provides the critical value we are looking for. 
Note that while first positive minimum $\kappa_{\mathrm{min}}\sim d/2$ for large $d$, the position of the first maximum (see Fig.~\ref{fig:Sk}) scales as $\kappa_{\mathrm{max}} \sim \sqrt{2d}$, as  can be verified expanding the hypergeometric functions for large $d$ around $\kappa^2 = 2d$.

\subsection{Numerical and Asymptotic Analysis of $\beta_c(d)$}

Since the values of the local minima of $S(\kappa)$ 
depend on $\beta_c(d)$, we first study $\beta_c(d)$ numerically and subsequently asymptotically. We return to analyzing the nonnegativity of the other 
local minima in Sec. \ref{sec:other-minima}.

We begin here by solving Eq.~\eqref{eq:beta_c} numerically with the approximation for $\kappa_{\mathrm{min}}$ in Eq.~\eqref{eq:kmin} and plot the solution in Fig.~\ref{fig:beta_c}.
\begin{figure}
    \centering
    \includegraphics[width=0.5\linewidth]{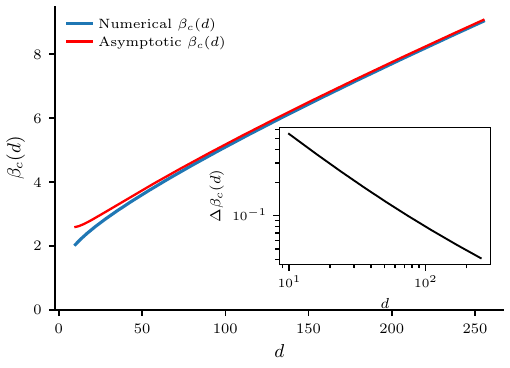}
    \caption{Numerical solution of Eq.~\eqref{eq:beta_c}, using Eq.~\eqref{eq:kmin} for $\kappa_{\mathrm{min}}$, for $c=2$ (blue), compared with the asymptotic expression in Eq.~\eqref{eq:asymptotic} (red). The inset shows that the difference between the two tends to zero as $d\to\infty$.}
    \label{fig:beta_c}
\end{figure}
According to Fig.~\ref{fig:beta_c}, $\beta_c(d)$ diverges with $d$, implying that one can choose any value of $\beta$ and have a well-defined structure factor in the limit $d \to \infty$. In particular, for a prescribed exponent $\beta$, the procedure produces admissible candidate pair statistics in dimensions $d \geq \ceil{\beta_c^{-1}(\beta)}$ (where $\ceil{\cdot}$ is the ceiling operator).
This provides a way of approaching the exponentially improved lower bound of Ref.~\cite{To06b}.

We now derive the asymptotic expansion of $\beta_c(d)$ starting from Eq.~\eqref{eq:beta_c}. We start by expanding $\kappa_{\mathrm{min}}(d)$ for large $d$, using Eq.~\eqref{eq:kmin}
\begin{equation}
    \kappa_{\mathrm{min}}(d) = \frac{d}{2} + \frac{a_1 \sqrt[3]{d}}{\sqrt[3]{2}}+\frac{c-2}{e^{c/2}-2}+ O(d^{-1/3})
\end{equation}
where $a_1$ satisfies $\operatorname{Ai}(-2^{1/3}a_1)=0$. We now expand the right-hand side of Eq.~\eqref{eq:beta_c}. To do so, we will adopt the uniform expansion of the Bessel functions
\begin{equation}
    J_{d/2}(d/2 + a_1 (d/2)^{1/3} + \tilde{c} + O(d^{-1/3})) = -2^{2/3} \, \tilde{c} \, \operatorname{Ai}'(-2^{1/3}a_1)(d/2)^{-2/3} + O(d^{-4/3}),
\end{equation}
\begin{equation}
    J_{d/2-1}(d/2 + a_1 (d/2)^{1/3} + \tilde{c} + O(d^{-1/3})) = -2^{2/3} (1+\tilde{c}) \operatorname{Ai}'(-2^{1/3}a_1)(d/2)^{-2/3} + O(d^{-4/3}),
\end{equation}
where we used that $\operatorname{Ai}(-2^{1/3}a_1)=0$ to remove the $O(d^{-1/3})$ contribution and defined $\tilde{c} = (c-2)/(e^{c/2}-2)$. For the asymptotic branch written below, we take $c>2\log 2$, so that the prefactor and the logarithm in Eq.~\eqref{eq:asymptotic} are positive and real, respectively. In addition, the $O(d^{-1})$ is removed by the explicit form of the $O(d^{-1/3})$ in the argument, which, however, we will not use, as we will focus on the leading term.
Using these expansions, the right-hand side of Eq.~\eqref{eq:beta_c} becomes
\begin{multline}
    \kappa_{\mathrm{min}}(d)^{d/2}\left[ \frac{J_{d/2}(\kappa_{\mathrm{min}}(d) \sigma)}{\sigma^{d/2}} - \frac{\kappa_{\mathrm{min}}(d)}{d(1+d^{-\beta_c})}J_{d/2-1}(\kappa_{\mathrm{min}}(d)) \right]^{-1} =\\
    =\frac{2^{-\frac{d}{2}-\frac{1}{3}} d^{-\beta_c +\frac{d}{2}+\frac{2}{3}} \left(d^{\beta_c}+1\right) \exp \left(\frac{1}{2} \left(2^{2/3} a_1
   \sqrt[3]{d}+c\right)+\frac{c-2}{e^{c/2}-2} + O(d^{-1/3})\right)}{\operatorname{Ai}'(-2^{1/3}a_1) \left(e^{c/2}-2\right)} (1+O(d^{-2/3})).
\end{multline}
Using Stirling's formula $\Gamma(1+d/2)=\sqrt{\pi d}\left(d/(2e)\right)^{d/2}(1+O(d^{-1}))$, we obtain
\begin{equation}
    d^{\beta_c} = \frac{2^{-\frac{d}{2}-\frac{1}{3}} d^{\frac{1}{6}} \exp \left(\frac{1}{2} \left(2^{2/3} a_1 \sqrt[3]{d}+c+d\right)+\frac{c-2}{e^{c/2}-2}\right)}{\sqrt{\pi }
   \operatorname{Ai}'(-2^{1/3}a_1) \left(e^{c/2}-2\right)}(1+O(d^{-1/3}))
\end{equation}
Taking the logarithm on both sides and dividing by $\log d$, we get the desired asymptotic solution
\begin{equation}
\label{eq:asymptotic}
\begin{aligned}
    \beta_c(d)={}&\frac{d(1-\log 2)}{2\log d}
    +\frac{a_1 d^{1/3}}{2^{1/3}\log d}+\frac{1}{6}\\
    &+\frac{-\log\!\left(2^{1/3}\sqrt{\pi}\,
    \operatorname{Ai}'(-2^{1/3}a_1)(e^{c/2}-2)\right)
    +\frac{c-2}{e^{c/2}-2}+\frac{c}{2}}{\log d}
    +o\!\left((\log d)^{-1}\right).
\end{aligned}
\end{equation}
The comparison between the numerical solution and the asymptotic result is shown in Fig.~\ref{fig:beta_c}.

Moreover, if one allows $\beta$ to depend on $d$, one can get a lower bound of the form 
\begin{equation}
    \phi_{\mathrm{max}} \geq \phi^* =  \frac{d^{\beta_c(d)}+1}{2^d}\left(1+\frac{c}{d}\right)^{-d} \sim \frac{e^{\frac{c-2}{e^{c/2}-2}-\frac{c}{2}}}{\sqrt[3]{2} \sqrt{\pi } \operatorname{Ai}'(-2^{1/3}a_1)
   \left(e^{c/2}-2\right)}\frac{d^{1/6} e^{a_1 (d/2)^{1/3}}}{2^{[3-\log_2(e)]d/2}} ,
\end{equation}
which coincides with the Torquato--Stillinger lower bound in Eq.~\eqref{eq:lower_bound} up to a $c$-dependent multiplicative constant. This procedure, therefore, allows for obtaining a complete family of polynomially improved lower bounds approaching the Torquato--Stillinger result.

In Fig.~\ref{fig:Sk}, we show the structure factors for different values of $d > \beta_c^{-1}(\beta)$ for selected large dimensions between 33 and 300. These structure factors satisfy the two-point nonnegativity conditions required of disordered packings. Their approach toward unity away from the constrained small-$\kappa$ and contact regimes illustrates decorrelation as the dimension increases.
\begin{figure}
    \centering
    \includegraphics[width=0.5\linewidth]{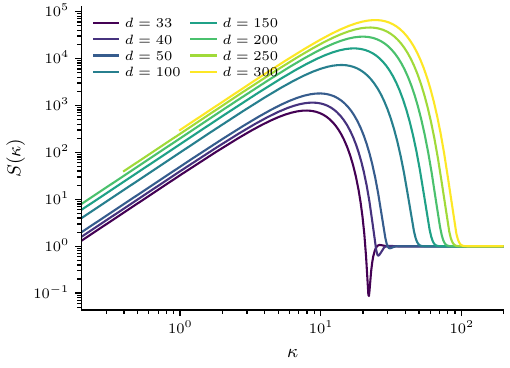}
    \caption{Structure factor for several dimensions, with $\beta=3$ and $c=2$. For these parameters, $\ceil{\beta_c^{-1}(3)}=33$; the figure shows that $S(\kappa)>0$ for $\kappa>0$ starting at $d=33$.}
    \label{fig:Sk}
\end{figure}

\subsection{All other minima of $S(\kappa)$ are nonnegative}
\label{sec:other-minima}

In the aforementioned optimization procedure, we constrained the first positive local minimum at $\kappa = \kappa_{\mathrm{min}}$ in $S(\kappa)$ to be nonnegative. 
We now argue that all the other local minima at $\kappa=\kappa^{(n)}_{\mathrm{min}}$ satisfy $S(\kappa^{(n)}_{\mathrm{min}}) \geq S(\kappa_{\mathrm{min}})$. Although we do not provide a uniform asymptotic treatment of the crossover between the Airy and Debye regimes, the potentially dangerous minima of $S(k)$ are the first few oscillations after the first positive minimum; subsequent minima are progressively suppressed by the decreasing Bessel envelope until the structure factor approaches unity in the limit $\kappa\infty$. Indeed,
in the large-$\kappa$ regime, it scales as
\begin{equation}
    S(\kappa) \sim 1 + \frac{2^{\frac{d-1}{2}} d^{\beta } k^{\frac{1}{2}-\frac{d}{2}} \Gamma \left(\frac{d}{2}\right)}{\sqrt{\pi }} \cos \left(k-\frac{1}{4}
   \pi  (d-1)\right),
\end{equation}
in agreement with the known scaling found in Ref.~\cite{To21c}. Thus, constraining the first positive minimum is sufficient for the asymptotic nonnegativity considered here.

\subsubsection{Airy scaling: $\kappa = d/2 + s (d/2)^{1/3}$, $s>a_1$}

The minimum at $\kappa = \kappa_{\mathrm{min}} \sim d/2 + a_1 (d/2)^{1/3}$ corresponds to the Airy-scaling region of the Bessel functions that enter in the expression for the structure factor [cf. \ref{eq:struc_fact}]. 
Here we consider the next several local minima that are located at $\kappa = d/2 + s (d/2)^{1/3}$ with $s>a_1$.
In this regime, we can expand the Bessel functions as
\begin{equation}
J_{d/2}(\sigma (d/2 + s (d/2)^{1/3})) \sim 2^{1/3}(d/2)^{-1/3}\operatorname{Ai}(-2^{1/3}s) - 2^{-1/3}c\,(d/2)^{-2/3}\operatorname{Ai}'(-2^{1/3}s),
\end{equation}
\begin{equation}
J_{d/2-1}(d/2 + s (d/2)^{1/3}) \sim 2^{1/3}(d/2)^{-1/3}\operatorname{Ai}(-2^{1/3}s) - 2^{2/3}(d/2)^{-2/3}\operatorname{Ai}'(-2^{1/3}s).
\end{equation}
Substituting these expressions in Eq.~\eqref{eq:struc_fact}, expanding at large $d$ and using the expression for $\beta_c$ in Eq.~\eqref{eq:beta_c}, we find
\begin{equation}
\begin{aligned}
 &\left|S\!\left(d/2+s(d/2)^{1/3}\right)-1\right|\\
 &\quad\sim
 \left|\exp\!\left(\frac{d^{1/3}(a_1-s)}{2^{1/3}}\right)
 \frac{d^{1/3}\operatorname{Ai}(-2^{1/3}s)
 e^{(c-2)/(e^{c/2}-2)}}{2^{2/3}\operatorname{Ai}'(-2^{1/3}a_1)}\right|\\
 &\quad\leq
 \left|\exp\!\left(\frac{d^{1/3}(a_1-s)}{2^{1/3}}\right)
 \frac{d^{1/3}\operatorname{Ai}(-2^{1/3}s)}
 {\operatorname{Ai}'(-2^{1/3}a_1)}\right|<1 .
\end{aligned}
\end{equation}
The right-hand side decays as a stretched exponential in $d$ for fixed $s>a_1$, and it is smaller than unity for any $s>a_1$, thus proving that the minima after the first are less deep, guaranteeing that $S(\kappa) \geq 0$ for $\kappa\leq d$.

\subsubsection{Intermediate Airy-Debye scaling: $d^{1/3}\ll \kappa-d/2\ll d$}

Here we provide analysis that shows the positivity of $S(\kappa)$ 
[cf. \ref{eq:struc_fact}] in the intermediate regime, as defined by
\begin{equation}
    \frac{d}{2}+s\left(\frac{d}{2}\right)^{1/3}
    \leq \kappa\leq z\frac{d}{2}
\end{equation}
for fixed $s>a_1$ and $z>1$. The standard uniform oscillatory-region estimate for Bessel functions is
\begin{equation}
    |J_m(x)|
    \leq C(x^2-m^2)^{-1/4},
    \qquad
    x\geq m+sm^{1/3}.
\end{equation}
It follows from the uniform Airy expansion of $J_m(mz)$ and the bound
$|\operatorname{Ai}(-t)|\leq Ct^{-1/4}$. Applied to the Bessel functions in
Eq.~\eqref{eq:struc_fact}, it gives
\begin{equation}
    |J_{d/2}(\sigma\kappa)|
    +|J_{d/2-1}(\kappa)|
    \leq
    C\left(\kappa^2-\frac{d^2}{4}\right)^{-1/4}.
\end{equation}

We consider the limiting case $\beta=\beta_c(d)$, since the coefficients in
$S(\kappa)-1$ decrease in absolute value when $\beta$ decreases. Using
Eq.~\eqref{eq:asymptotic} and Stirling's formula, we obtain
\begin{equation}
    \frac{c_2(d)}{\kappa^{d/2-1}},\,\frac{c_1(d)}{(\sigma\kappa)^{d/2}}
    \sim \left(\frac{d}{2}\right)^{2/3} \exp\!\left[ a_1\left(\frac{d}{2}\right)^{1/3} -\frac{d}{2}\log\!\left(\frac{2\kappa}{d}\right) \right]
\end{equation}
Here, the bounded factors $\sigma^{-d/2}$ and $2\kappa/d$ have been absorbed
into the constants. Substitution into Eq.~\eqref{eq:struc_fact} therefore yields
\begin{equation}
    |S(\kappa)-1| \leq
    C\left(\frac{d}{2}\right)^{2/3}
    \left(\kappa^2-\frac{d^2}{4}\right)^{-1/4}
    \exp\!\left[
        a_1\left(\frac{d}{2}\right)^{1/3}
        -\frac{d}{2}\log\!\left(\frac{2\kappa}{d}\right)
    \right].
\end{equation}
The right-hand side decreases with $\kappa$ and, at
$\kappa=d/2+s (d/2)^{1/3}$, it is
\begin{equation}
    |S(\kappa)-1|
    \leq
    \exp\!\left[
        -(s-a_1)\left(\frac{d}{2}\right)^{1/3}+O(\log d)
    \right].
\end{equation}
Thus, it is uniformly stretched exponentially small throughout the
intermediate regime, and hence $S(\kappa)>0$ for all sufficiently large $d$.

\subsubsection{Debye scaling: $\kappa = z \,d/2$, $z>1$}

We are left with the study of the nonnegativity of the local minima of $S(\kappa)$ for $\kappa>d$ for large $\kappa$. We can expand the Bessel functions in Debye’s regime $\kappa = z d/2$, with $z>1$. In this regime, we have
\begin{equation}
J_{d/2}(zd/2 ) \sim \left(\frac{2}{\pi d/2}\right)^{1/2} \frac{1}{(z^2-1)^{1/4}} \cos\left[ d/2 \left(\sqrt{z^2-1}-\arccos\frac1z\right) -\frac{\pi}{4} \right],
\end{equation}
\begin{equation}
J_{d/2-1}(z d/2) \sim \left(\frac{2}{\pi d/2}\right)^{1/2} \frac{1}{(z^2-1)^{1/4}} \cos\left[ d/2\left(\sqrt{z^2-1}-\arccos\frac1z\right)+\arccos\frac1z-\frac{\pi}{4} \right].
\end{equation}
We can now bound $|S(\kappa)-1|$ as
\begin{equation}
    |S(\kappa)-1| \leq \bigg|- c_1(d) \frac{J_{d/2}(\kappa \sigma)}{(\kappa\sigma)^{d/2}} + c_2(d) \frac{J_{d/2-1}(\kappa)}{\kappa^{d/2-1}}\bigg| \leq \bigg| c_1(d) \frac{J_{d/2}(\kappa \sigma)}{(\kappa\sigma)^{d/2}} \bigg| + \bigg| c_2(d) \frac{J_{d/2-1}(\kappa)}{\kappa^{d/2-1}} \bigg|
\end{equation}
For $\kappa = z d/2$ this becomes, for large $d$
\begin{align}
    |S(\kappa)-1| &\leq  \left(\frac{2}{\pi d/2}\right)^{1/2} \frac{1}{((\sigma z)^2-1)^{1/4}} \frac{c_1(d)}{(d/2 z \sigma)^{d/2}}  +  \left(\frac{2}{\pi d/2}\right)^{1/2} \frac{1}{(z^2-1)^{1/4}} \frac{c_2(d)}{(z d/2)^{d/2-1}} \\
    & = \frac{2^d}{\sqrt{\pi }} \sqrt{\frac{1}{d}} \Gamma \left(\frac{d}{2}+1\right) \left(\frac{2 \left(d^{\beta }+1\right) (z (c+d))^{-d/2}}{\sqrt[4]{\frac{z^2
   (c+d)^2}{d^2}-1}}+\frac{z d^{\beta } (d z)^{-d/2}}{\sqrt[4]{z^2-1}}\right)
   \\
   & \sim \frac{ \left((z+2) d^{\beta }+2\right)}{\sqrt[4]{z^2-1}} \, 2^de^{-\frac{1}{2} d (\log (z)+1+\log (2))}.
\end{align}
The function on the right-hand side decays exponentially with $d$ for any fixed $z$, thus proving that $S(\kappa)$ is exponentially close to unity for large $d$. Together with the Airy-regime analysis above, this shows asymptotically that if $S(\kappa_{\mathrm{min}}) \geq 0$, then $S(\kappa) \geq 0$ for all $\kappa>0$ in the regimes considered.

\section{Alternative Proof of the Torquato-Stillinger Exponential \\Asymptotic Form}
\label{sec:proof}

Samorodnitsky used Delsarte's linear programming approach to derive bounds on the maximal size $A(n, {\cal D})$ of a binary error-correcting
code of length $n$ and distance $\cal D$, or, alternatively, for the best packing of
spheres in $d$-dimensional Hamming space \cite{Sa01}.
Motivated by these results, we derive a lower bound on the Cohn--Elkies LP upper bound that has the same asymptotic form as the Torquato--Stillinger conjectured exponential improvement of Minkowski's lower bound.
The proof can physically be interpreted as a type of ``diffraction-limit" argument. A Cohn–Elkies test function is a positive-definite radial certificate that must become nonpositive outside the hard-core radius. Its low Fourier modes (small wavenumbers), however, are too smooth and, after radial weighting, too monotonic to produce the required sign change. Consequently, the sign change must be generated by the high Fourier modes (large wavenumbers), which create a compensating negative dip near contact. In high dimensions, however, the decay of the Bessel functions strongly suppresses the contribution of these high modes, making them increasingly inefficient. As a result, the necessary high-wavenumber correction prevents the value of the Cohn–Elkies certificate from decaying faster than the Torquato–Stillinger exponential scale.

\begin{theorem} \label{theorem:bound}
Let $d \ge 2$, and let $f:\mathbb R^d\to\mathbb R$ be a radial Schwartz function, and suppose that $f(r) \le 0$ for all $r \ge 1$, $\tilde{f}(q) \ge 0$ for all $q \ge 0$, and $\tilde{f}(0)>0$.
If $Q>0$ is such that for $0 \le q \le Q$, the function $r \mapsto r^{d-1} B_d(rq)$ is weakly increasing on $[0,1]$, then
\begin{equation}
\vol(B_{1/2}^d)\frac{f(0)}{\tilde{f}(0)}
\ge
\frac{(\pi Q)^{d/2-1}}{2^{d+1} d \Gamma(d/2)}.
\end{equation}
\end{theorem}

\begin{proof}
We use a cutoff wavenumber $Q$ to decompose $f$ [cf. Eq.~(\ref{f-driect})] into low- and high-wavenumber contributions and then optimize the cutoff.
We therefore write $f = f_{\le Q} + f_{\ge Q}$ where
\begin{equation}
f_{\le Q}(r)
=
\int_0^Q \tilde{f}(q) B_d(rq) \, \rmd\nu_d(q)
\end{equation}
is the low-wavenumber contribution, and
\begin{equation}
f_{\ge Q}(r) = \int_Q^\infty \tilde{f}(q) B_d(rq) \, \rmd\nu_d(q)
\end{equation}
is the high-wavenumber contribution.
The role of $Q$ is to separate those Fourier modes whose weighted radial kernels $r^{d-1}B_d(rq)$ are monotone on the hard-core interval from the complementary modes. Indeed, by assumption, for every $0\le q\le Q$ these kernels are weakly increasing on $[0,1]$. Since $\tilde f(q)\ge 0$, the low-wavenumber radial density associated with $f_{\le Q}$ is therefore monotone on $[0,1]$. Consequently, the low-wavenumber part cannot by itself create the negative dip required by the sign condition $f(1)\le 0$; the necessary cancellation at contact must come from the high-wavenumber component $f_{\ge Q}$.
Let $F$, $F_{\le Q}$, and $F_{\ge Q}$ be the corresponding radial density functions
\begin{equation}
F(r)=\omega_d r^{d-1}f(r),
\qquad
F_{\le Q}(r)=\omega_d r^{d-1}f_{\le Q}(r),
\qquad
F_{\ge Q}(r)=\omega_d r^{d-1}f_{\ge Q}(r),
\end{equation}
where $\omega_d = d\,\vol(B_1^d)$ is the surface area of $S^{d-1}$ (so $\rmd\nu_d(r) = \omega_d r^{d-1} \, \rmd r$).

The strategy behind the proof is to use the nonnegativity of $\tilde{f}(q)$ and the monotonic increase of the kernels $r^{d-1} B_d(rq)$ to show that $F_{\le Q}(r)$ is monotonically increasing in $[0,1]$. 
Moreover, the sign condition $f(r) \le 0$ for all $r \ge 1$ forces $f(1)\le 0$, and consequently $F_{\ge Q}(r)$ must create a negative dip near $r=1$ large enough to cancel the low-wavenumber contribution from $F_{\le Q}(r)$. In particular, $|F_{\ge Q}|$ must be at least $\tilde{f}(0)/2$ somewhere in $[0,1]$: see Fig.~\ref{fig:wavenumber_fig} for an illustration.
\begin{figure}
    \centering
    \includegraphics[width=0.45\linewidth]{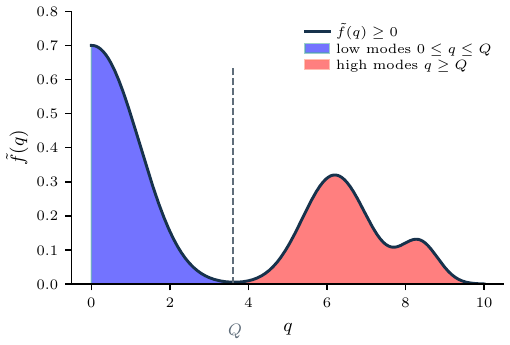}
    \includegraphics[width=0.45\linewidth]{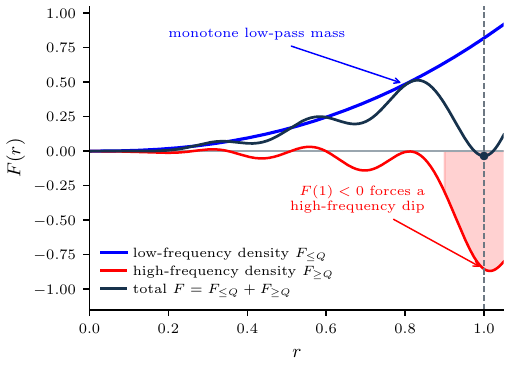}
    \caption{(Left) Illustration of the Fourier transform of $f$, $\tilde{f}(q) \ge 0$, split into low ($0 \le q \le Q$, blue) and high modes ($q \ge Q$, red). (Right) Plot of the radial density $F(r)$, split into low-wavenumber (blue) and high-wavenumber (red) contributions, such that their sum $F = F_{\le Q} + F_{\ge Q}$ is negative at $r=1$.}
    \label{fig:wavenumber_fig}
\end{figure}
The intuitive picture is that low and high modes have positive spectral masses, as $\tilde{f}(q)\ge0$; consequently, the cancellation at $r=1$ comes from oscillation of the Bessel kernel.
Combining this lower bound with a straightforward upper bound for $F_{\ge Q}$ in terms of $f(0)$ will yield the desired result.

The monotonicity of $F_{\le Q}$ on $[0,1]$ follows from the hypothesis on $Q$ and the inequality $\tilde{f}(q)\ge0$ for $0\leq q\leq Q$. Indeed,
\begin{equation}
F_{\le Q}(r)
=
\omega_d r^{d-1}
\int_0^Q \tilde{f}(q)B_d(rq)\,\rmd\nu_d(q),
\end{equation}
and each function $r \mapsto r^{d-1}B_d(rq)$ appearing in the integral is weakly increasing on $[0,1]$.

We now show that $|F_{\ge Q}|$ must be large somewhere in $[0,1]$, specifically at least $\tilde{f}(0)/2$. Since $f(1)\le0$, we have
\begin{equation}
F_{\le Q}(1) + F_{\ge Q}(1) = F(1)\le0.
\end{equation}
We have $F_{\le Q}(1) \ge F_{\le Q}(0)=0$ by monotonicity, and hence $F_{\ge Q}(1) \le 0$.
Let
\begin{equation}
\varepsilon=\max_{0\le r\le1}|F_{\ge Q}(r)|.
\end{equation}
Then
\begin{equation}
F_{\le Q}(1)\le -F_{\ge Q}(1)\le \varepsilon,
\end{equation}
and since $F_{\le Q}$ is weakly increasing, it follows that
\begin{equation}
F_{\le Q}(r)\le \varepsilon
\qquad\text{for }0\le r\le1.
\end{equation}
Thus,
\begin{equation}
F(r)=F_{\le Q}(r)+F_{\ge Q}(r)\le2\varepsilon
\qquad\text{for }0\le r\le1.
\end{equation}
However, the inequality $f(r) \le 0$ for $r \ge 1$ implies that
\begin{equation}
\tilde{f}(0) = \int_0^\infty F(r) \, \rmd r \le \int_0^1 F(r) \, \rmd r \le 2 \varepsilon.
\end{equation}
It follows that $\varepsilon \ge \tilde{f}(0)/2$, i.e.,
\begin{equation}
\max_{0\le r\le1}|F_{\ge Q}(r)|
\ge
\frac{\tilde{f}(0)}2.
\end{equation}

The next step is to bound $F_{\ge Q}$ from above. Set
\begin{equation}
K_Q(r)=\max_{q\ge Q}|B_d(rq)|.
\end{equation}
Then for $0 <r\le1$,
\begin{equation}
|f_{\ge Q}(r)|
\le
\int_Q^\infty \tilde{f}(q)|B_d(rq)|\,\rmd\nu_d(q)
\le
K_Q(r)\int_Q^\infty \tilde{f}(q)\,\rmd\nu_d(q)
\end{equation}
since $\tilde{f}(q) \ge 0$.
We have
\begin{equation}
f(0)
=
\int_0^\infty \tilde{f}(q)\,\rmd\nu_d(q),
\end{equation}
and hence
\begin{equation}
\int_Q^\infty \tilde{f}(q)\,\rmd\nu_d(q)\le f(0).
\end{equation}
Therefore
\begin{equation}
|F_{\ge Q}(r)|
=
\omega_d r^{d-1}|f_{\ge Q}(r)|
\le
\omega_d r^{d-1}K_Q(r)f(0).
\end{equation}
Taking the supremum over $0< r\le1$, we obtain
\begin{equation}
\frac{\tilde{f}(0)}2
\le
f(0)\sup_{0< r\le1}\omega_d r^{d-1}K_Q(r).
\end{equation}

To bound $K_Q(r)$ from above, we write
\begin{equation}
K_Q(r)= \Gamma(d/2)\max_{q\ge Q} \frac{|J_{d/2-1}(2\pi rq)|}{(\pi rq)^{d/2-1}}
\end{equation}
and use the inequality $|J_\nu(x)| \le 1$ for $\nu \ge 0$ and $x \in \R$ (see equation~(10) in \S13.42 of \emph{A Treatise on the Theory of Bessel Functions} by G.~N.~Watson, second edition) to obtain
\begin{equation}
\label{eq:upperbound_Kr}
K_Q(r) \le \frac{\Gamma(d/2)}{(\pi r Q)^{d/2-1}}.
\end{equation}
Figure~\ref{fig:bessel_kernel} compares $|B_d(rq)|$ for $q/Q=1$, $1.5$, and $2$ with the uniform upper bound in Eq.~\eqref{eq:upperbound_Kr}; the logarithmic vertical scale displays the increasingly strong suppression of the oscillatory kernel as either $r$ or $q$ increases.
Thus,
\begin{equation}
\frac{\tilde{f}(0)}2
\le
f(0)\sup_{0< r\le1} \frac{\omega_d \Gamma(d/2) r^{d/2}}{(\pi Q)^{d/2-1}} = f(0) \frac{\omega_d \Gamma(d/2)}{(\pi Q)^{d/2-1}}.
\end{equation}
Finally, we obtain the conclusion of the theorem by using $\vol(B^d_{1/2}) = \vol(B^d_1)/2^d$ and $\omega_d = d\,\vol(B_1^d)$.
\end{proof}
\begin{figure}
    \centering
    \includegraphics[width=0.5\linewidth]{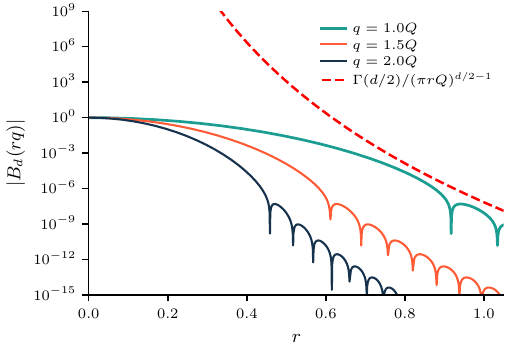}
    \caption{Comparison between $|B_d(r q)|$ for $d=70$ and $Q=7$ for different values of $q$ and the upper bound in Eq.~\eqref{eq:upperbound_Kr}.}
    \label{fig:bessel_kernel}
\end{figure}

Using a stronger bound than $|J_\nu(x)| \le 1$ at the end of the proof would slightly improve the bound, but only by subexponential factors.

Let $j_{\nu,1}$ be the first positive root of the Bessel function $J_\nu$. It is a standard fact that $j_{\nu,1}$ is a strictly increasing function of $\nu$ for $\nu > -1$.

\begin{lemma}
\label{th:lemma}
If $d \ge 3$, then the function
\begin{equation}
r\mapsto r^{d-1}B_d(rq)
\end{equation}
is weakly increasing on $[0,1]$ whenever
\begin{equation}
0\le q\le \frac{j_{d/2-2,1}}{2\pi}.
\end{equation}
\end{lemma}

\begin{proof}
The case $q=0$ is immediate, since 
$r^{d-1}B_d(0)=r^{d-1}$, so we assume $q>0$.

Let $\nu = d/2 - 1$ and $x = 2 \pi q r$. Then $r^{d-1}B_d(rq)$
is equal to
$x^{\nu+1}J_\nu(x)$ times a positive factor that depends only on $d$ and $q$.
Thus, it suffices to prove that
\begin{equation}
x\mapsto x^{\nu+1}J_\nu(x)
\end{equation}
is increasing for
\begin{equation}
0\le x\le j_{\nu-1,1}.
\end{equation}

Differentiating gives
\begin{equation}
\frac{\rmd}{\rmd x}\left(x^{\nu+1}J_\nu(x)\right)
=
x^\nu(xJ_{\nu-1}(x)+J_\nu(x)).
\end{equation}
Now $J_{\nu-1}(x)>0$ for $0 < x < j_{\nu-1,1}$, and $J_{\nu}(x)>0$ for $0 < x < j_{\nu,1}$. Because $\nu-1 > -1$, we have $j_{\nu-1,1} < j_{\nu,1}$ as noted above, and thus
\begin{equation}
\frac{\rmd}{\rmd x}\left(x^{\nu+1}J_\nu(x)\right) > 0 \quad \text{for $0 < x < j_{\nu-1,1}$},
\end{equation}
as desired.
\end{proof}

We can now choose the largest cutoff allowed by Lemma~\ref{th:lemma}, namely
$Q_*= \frac{j_{d/2-2,1}}{2\pi}$.
Substituting this value of $Q$ into Theorem~\ref{theorem:bound} gives, for every admissible radial Cohn–Elkies test function $f$,
\begin{equation}
    \operatorname{vol}(B^d_{1/2})\frac{f(0)}{\tilde f(0)} \ge L_d:= \frac{\left(j_{d/2-2,1}/2\right)^{d/2-1}} {2^{d+1}d\,\Gamma(d/2)} .
\end{equation}
Thus, $L_d$ is a lower bound on the value of the Cohn–Elkies upper-bound functional itself: no admissible radial dual certificate can certify an upper bound asymptotically smaller than $L_d$.
Using the standard first-zero asymptotic
$j_{\nu,1}=\nu+a_1\nu^{1/3}+O(\nu^{-1/3})$, $a_1\simeq 1.8557571$, together with Stirling’s formula, we obtain
\begin{equation}
    L_d = \exp\!\left(O(d^{1/3})\right) 2^{-\frac{3-\log_2 e}{2}d} = \exp\!\left(O(d^{1/3})\right) 2^{-0.7786524795\,d}.
\end{equation}
This is precisely the exponential rate appearing in the Torquato–Stillinger conjectural lower bound. The present argument should therefore be interpreted as a dual obstruction: the radial Cohn–Elkies LP cannot rule out packings at the Torquato–Stillinger exponential scale.

\section{Robustness of the Torquato-Stillinger ``Universality" Class}
\label{sec:universality}

Evidence is mounting that the Torquato-Stillinger exponential rate with exponent
$0.7786524795\ldots$ is ``universal" in the sense that it can be obtained 
using completely different methods from the LP formulations \cite{Ed25}, 
and as well as the LP lower bound using different test $g_2$ functions (including relaxation of the hyperuniformity constraint) \cite{Sc08}. The present paper reports
that the same exponential improvement in the lower bound can be approached in the high-$d$
limit using polynomial improvements. We also showed that the Cohn–Elkies dual linear programming upper bound formulation cannot asymptotically exclude
packings with the Torquato–Stillinger density scalings.

To explore further the robustness of the $0.7786524795\ldots$, we examined the use of  a broad
class of other radial pair-correlation test functions $g_2(r)$  in the Torquato--Stillinger lower-bound program for high-dimensional sphere packings, including the pure hard-core step function, a step function with a contact delta function, the original step-plus-delta-plus-gap ansatz, additional positive delta shells, partially filled gaps, smooth near-contact peaks, finite monotone staircases, positive continuum mixtures of steps, shrinking-scale perturbations $r-1=O(d^{-\gamma})$, and finite nonmonotone sequences of depleted and enhanced radial bands. Although several of these families improve finite-dimensional packing fractions or their polynomial and stretched-exponential prefactors, the analysis shows that none improves the leading Torquato--Stillinger scaling
\[
\phi_*
\sim
2^{-\alpha_{\mathrm{TS}}d+o(d)},
\qquad
\alpha_{\mathrm{TS}}
=
\frac{3-\log_2 e}{2}
=
0.7786524795\ldots .
\]
The reason is a universal feature of the high-dimensional radial Fourier transform. The structure-factor condition $S(k)\geq 0$ is controlled by Bessel functions whose order is approximately $d/2$. For every fixed-complexity test function whose correlations remain confined to a shrinking neighborhood of contact, the first restrictive minimum of $S(k)$ remains near the Bessel turning point
\[
k=\frac{d}{2}+O(d^{1/3}),
\]
where the Bessel functions are described by Airy-function asymptotics. These asymptotics fix the coefficient of $d$ in $\log\phi_*$, while changes in gap widths, contact numbers, shell positions, or finitely many radial amplitudes affect only algebraic factors or stretched-exponential corrections such as $\exp[O(d^{1/3})]$. Positive outer shells face an additional low-wavenumber moment obstruction: the simultaneous requirements of hyperuniformity and $S''(0)\geq 0$ restrict the density to the familiar $O(d\,2^{-d})$ scale unless $g_2(r)$ contains a genuinely depleted region with $g_2(r)<1$. Positive staircases can cancel the leading Airy contribution, but positivity and convexity prevent them from cancelling the next independent Airy mode. Signed, nonmonotone finite-band functions remove this particular convexity obstruction, yet cancelling any fixed number of Airy terms still leaves the same underlying exponential envelope. To escape the Torquato--Stillinger universality class, the active Fourier constraint would have to be displaced from the turning-point region into a \emph{Debye-frequency interval}, by which we mean a finite range of wavenumbers
\[
k=z\,\frac{d}{2},
\qquad
1+\epsilon\leq z\leq z_*,
\]
where $\epsilon>0$ and $z_*>1$ remain fixed as $d\to\infty$. Physically, this is a band of wavelengths whose wavenumbers scale linearly with dimension and lie a finite fractional distance above the turning point, rather than only $O(d^{1/3})$ above it. In this regime, the Bessel kernels are described by oscillatory Debye asymptotics, and moving the first active minimum to a fixed $z_*>1$ would change the coefficient multiplying $d$ and hence genuinely improve the packing exponent. Achieving such a displacement appears to require several independently tunable depleted and enhanced radial bands that grow with dimension, so that the test function suppresses an entire $O(d)$-wide Fourier interval rather than merely shifting or cancelling a finite set of turning-point minima.

Thus, the Torquato–Stillinger universality class is remarkably robust across a wide range of radial test pair-correlation functions. Escaping the TS universality class consequently requires a substantial qualitative change in the possible radial test functions: the number of independently adjustable radial features must grow with dimension so that the pair correlation suppresses an entire Debye-frequency interval $k=zd/2$, with $z>1$ fixed, rather than merely perturbing the neighborhood of the first turning-point minimum.

\section{Pair Correlations Improving on the Torquato--Stillinger Exponential Rate}
\label{sec:radau-construction}

We now describe an explicit family of radial test pair correlations
whose terminal packing fractions have exponent
$0.622556248918\ldots$, improving on the
Torquato--Stillinger exponent $0.7786524795\ldots$.
The construction uses Gauss--Radau quadrature and a number of positive
radial shells that grows linearly with the dimension.

Let $d\geq5$, put $\nu=d/2$, and set
\begin{equation}
    m=m_d=\left\lfloor\frac{d-1}{4}\right\rfloor.
    \label{eq:radau-order}
\end{equation}
Consider the probability measure
\begin{equation}
    \rmd\mu_\nu(t)=\nu t^{\nu-1}\rmd t,
    \qquad 0\leq t\leq1.
\end{equation}
Let $0<t_1<\cdots<t_m<1$ be the zeros of
\begin{equation}
    P_m^{(0,\nu)}(2t-1),
\end{equation}
where $P_m^{(0,\nu)}$ is the Jacobi polynomial.  The associated left
Gauss--Radau weights $q_0,q_1,\ldots,q_m$ are the unique positive numbers for
which
\begin{equation}
    \int_0^1p(t)\rmd\mu_\nu(t)
    =q_0p(0)+\sum_{j=1}^m q_jp(t_j)
    \label{eq:radau-quadrature}
\end{equation}
for every polynomial $p$ of degree at most $2m$.  Thus the construction is
completely explicit in terms of Jacobi-polynomial zeros.  More concretely, if
$u_0=0$, $u_j=t_j$, and
\begin{equation}
    \ell_i(t)=\prod_{\substack{0\leq l\leq m\\l\neq i}}
    \frac{t-u_l}{u_i-u_l},
\end{equation}
then
\begin{equation}
    q_i=\nu\int_0^1\ell_i(t)^2t^{\nu-1}\rmd t,
    \qquad
    q_0=\binom{\nu+m}{m}^{-2}.
    \label{eq:radau-weights}
\end{equation}

Define
\begin{equation}
    \lambda=q_0^{-1},\qquad
    \sigma=t_1^{-1/2},\qquad
    R_j=\sqrt{\frac{t_j}{t_1}},\qquad
    Z_j=\lambda q_j,
    \qquad
    \rho=\frac{\lambda}{v_d\sigma^d}.
    \label{eq:radau-parameters}
\end{equation}
The radial pair-correlation function is
\begin{equation}
    g_2(r)
    =\Theta(r-\sigma)
    +\frac{1}{\rho\omega_d}\sum_{j=1}^m
    \frac{Z_j}{R_j^{d-1}}\delta(r-R_j).
    \label{eq:radau-pair-measure}
\end{equation}
Here $R_1=1$, while $1<R_j<\sigma$ for $j\geq2$, so $g_2(r)$ satisfies the
hard-core condition.

To write its structure factor, define
\begin{equation}
    B_\nu(x)=\Gamma(\nu)\left(\frac{2}{x}\right)^{\nu-1}J_{\nu-1}(x),
    \qquad
    A_\nu(x)=\Gamma(\nu+1)\left(\frac{2}{x}\right)^\nu J_\nu(x),
\end{equation}
with their values at zero defined by continuity.  The structure factor
associated with Eq.~\eqref{eq:radau-pair-measure} is
\begin{equation}
    S(k)
    =\lambda\left[
      q_0+\sum_{j=1}^m q_jB_\nu(k\sigma\sqrt{t_j})
      -A_\nu(k\sigma)\right].
    \label{eq:radau-structure-factor}
\end{equation}

\begin{proposition}
\label{prop:radau-shell-family}
For every $d\geq5$, the nonnegative radial pair-correlation function
$g_2(r)$ in Eq.~\eqref{eq:radau-pair-measure} satisfies the hard-core
condition, and its structure factor is nonnegative for all wavenumbers and is
hyperuniform:
\begin{equation}
    S(k)\geq0\quad(k\geq0),
    \qquad S(0)=0.
    \label{eq:radau-structure-positivity}
\end{equation}
Moreover, $g_2(r)=1$ for $r\geq\sigma$.  Thus, the total correlation function
$h(r)=g_2(r)-1$ has compact support, so the pair correlations vanish
identically beyond a finite distance.  Finally, the
terminal packing fraction is
\begin{equation}
    \phi_{*,d}
    =\frac{\rho v_d}{2^d}
    =2^{-2\nu}\binom{\nu+m}{m}^2t_1^\nu
    =\left(\frac{3\sqrt{3}}{8}+o(1)\right)^d
    =2^{-(0.622556248918\ldots+o(1))d}.
    \label{eq:radau-density-rate}
\end{equation}
\end{proposition}

The condition needed for the structure-factor sign is $d\geq4m+1$, which is satisfied by the choice in
Eq.~\eqref{eq:radau-order}.
Figure~\ref{fig:radau-pair-structure} illustrates the pair correlations and
their corresponding structure factors in three representative dimensions.
\begin{figure}[!ht]
    \centering
    \includegraphics[width=\linewidth]{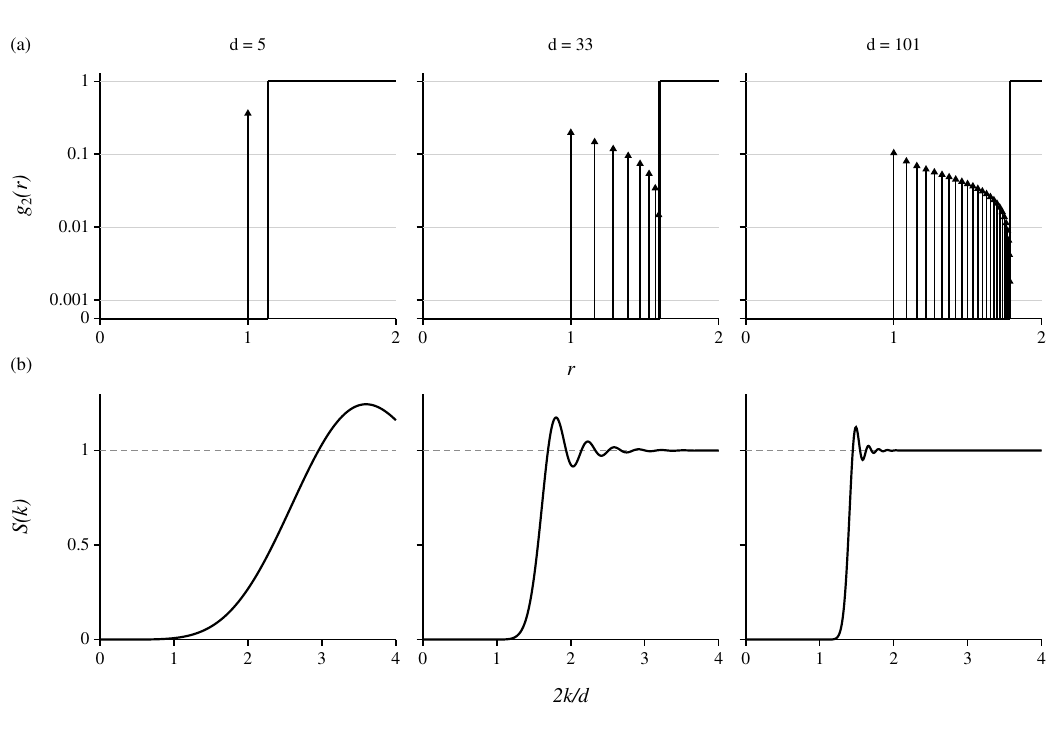}
    \caption{Radau pair-correlation functions $g_2(r)$ (top) and corresponding
    structure factors $S(k)$ (bottom) for $d=5,33,101$, ordered from left to
    right.  The delta functions in the top row are drawn as arrows whose
    heights equal their coefficients in Eq.~\eqref{eq:radau-pair-measure};
    positive values are shown on a logarithmic scale, with the horizontal axis
    representing zero.  In each dimension, $g_2(r)=1$ for $r\geq\sigma$.  The
    bottom row uses the scaled wavenumber $2k/d$, and the dashed line marks
    $S(k)=1$.}
    \label{fig:radau-pair-structure}
\end{figure}

\section{Pair Correlations with the Optimal Exponential Rate}
\label{sec:finite-band-construction}

Recently, OpenAI announced a result that determines the optimal high-dimensional exponential rate of the Cohn–Elkies linear-programming bound and thereby improves the rigorous upper bound on the maximal sphere-packing density~\cite{openai2026math}. The preceding sections address the complementary problem of obtaining conjectural lower bounds through the Torquato–Stillinger pair-correlation framework. Here we relate these two formulations and construct ordinary pair-correlation test functions at the optimal Cohn–Elkies exponential rate $2^{-(0.6044005\ldots+o(1))d}$.

The Torquato–Stillinger pair-correlation program is the dual of the Cohn–Elkies program, as already observed in Ref.~\cite{To06b}. In the generalized measure formulation of Ref.~\cite{Co22}, strong duality shows that their optimal values coincide: there is no duality gap. We show that this common optimal value can be approached arbitrarily closely, in every fixed dimension, by ordinary nonnegative radial functions $g_2(r)$ with $S(k)\geq 0$ and $g_2(r)=1$ beyond a finite radius. These functions can moreover be chosen to have finitely many radial bands. Thus the unrestricted ordinary pair-correlation program has the optimal Cohn–Elkies exponential rate, although the construction does not assert that a single finite-band function attains the optimum.

For unit-diameter spheres, let
\begin{equation}
 p_d=\inf_f\frac{f(0)}{\widetilde f(0)},
 \qquad
 \phi_{\mathrm{CE}}(d)=\frac{v_d}{2^d}p_d,
 \label{eq:finite-band-ce-value}
\end{equation}
where the infimum is over radial Schwartz functions satisfying
$f(r)\leq0$ for $r\geq1$, $\widetilde f(k)\geq0$, and
$\widetilde f(0)>0$. In this notation, the OpenAI result is
\begin{equation}
 \phi_{\mathrm{CE}}(d)=2^{-(\alpha_*+o(1))d},
 \qquad
 \alpha_*=\frac12\log_2\!\left(\frac{2\pi}{e}\right)
          =0.6044005\ldots.
 \label{eq:finite-band-ce-exponent}
\end{equation}

The identification of the Torquato--Stillinger program with the dual of
the Cohn--Elkies program is already explicit in Ref.~\cite{To06b}.
In the rigorous generalized formulation~\cite[Sec.~3]{Co22}, the dual
maximizes $\rho$ over nonnegative pair measures $\mu$ supported on
$r\geq1$, subject to
\begin{equation}
 T=\delta_0+\mu,
 \qquad \widetilde T-\rho\delta_0\geq0.
 \label{eq:finite-band-dual-measure}
\end{equation}
For an ordinary pair correlation, $\rmd\mu(\bm r)=\rho g_2(r)\rmd\bm r$,
and the last inequality becomes
\begin{equation}
 \widetilde T=\rho\delta_0+S(k)\,\rmd\bm k,
 \qquad S(k)=1+\rho\widetilde h(k)\geq0,
 \qquad h(r)=g_2(r)-1.
 \label{eq:finite-band-dual-structure}
\end{equation}
Thus the support and positivity conditions on $\mu$ give precisely the
hard-core condition and $g_2(r)\geq0$, while positivity of the remaining
Fourier measure gives $S(k)\geq0$. Strong duality identifies the supremum
in this generalized program with $p_d$. The construction below shows
that the same supremum can be approached using ordinary functions with
finitely many radial bands.

\subsection{A Family with Finitely Many Radial Bands}

Choose radii $1=r_0<r_1<\cdots<r_N=R$ and nonnegative heights
$b_1,\ldots,b_N$, and set
\begin{equation}
 g_2(r)=
 \begin{cases}
 0,&0\leq r<1,\\
 b_j,&r_{j-1}\leq r<r_j,\quad 1\leq j\leq N,\\
 1,&r\geq R.
 \end{cases}
 \label{eq:finite-band-pair}
\end{equation}
The number of bands, their radii, and their heights may depend on $d$.
Unlike the family in Sec.~\ref{sec:radau-construction}, this family
contains no delta functions.

Let $\chi_a(\bm r)$ denote the indicator function of the ball of radius
$a$. With the Fourier convention used in this paper,
\begin{equation}
 \widetilde\chi_a(k)=
 \begin{cases}
 v_da^d,&k=0,\\
 (a/k)^{d/2}J_{d/2}(2\pi ak),&k>0.
 \end{cases}
 \label{eq:finite-band-ball-transform}
\end{equation}
The structure factor associated with Eq.~\eqref{eq:finite-band-pair} is
therefore
\begin{equation}
 S(k)=1-\rho\widetilde\chi_R(k)
       +\sum_{j=1}^N z_j
          [\widetilde\chi_{r_j}(k)-\widetilde\chi_{r_{j-1}}(k)],
 \qquad z_j=\rho b_j.
 \label{eq:finite-band-structure}
\end{equation}
For fixed radii, this expression is linear in the nonnegative variables
$\rho,z_1,\ldots,z_N$.

\begin{proposition}
\label{prop:finite-band-optimum}
For each $d\geq1$ and every $0<\varepsilon<1$, there are finite choices
of $N,R,r_j,b_j$ and a number density $\rho$ such that
Eq.~\eqref{eq:finite-band-pair} satisfies the hard-core condition,
$g_2(r)\geq0$, and $S(k)\geq\tau>0$ for every $k\geq0$, for some
$\tau>0$, with
\begin{equation}
 (1-\varepsilon)p_d<\rho\leq p_d.
 \label{eq:finite-band-approximation}
\end{equation}
Consequently, the supremum of the terminal packing fractions of this
family is $\phi_{\mathrm{CE}}(d)$ in every fixed dimension. In particular,
one can choose a member in each dimension with
\begin{equation}
 \phi_d=\frac{\rho_dv_d}{2^d}
       =2^{-(\alpha_*+o(1))d}.
 \label{eq:finite-band-density-rate}
\end{equation}
\end{proposition}

Equation~\eqref{eq:finite-band-approximation} concerns approximation of
the optimal value; it does not assert that a single finite-band function
attains $p_d$. Taking any fixed $\varepsilon$, for example
$\varepsilon=1/2$, already gives Eq.~\eqref{eq:finite-band-density-rate},
since a constant factor does not change the exponential rate.

The parameters can be found by solving finite linear programs. For
completeness, a sufficient finite set of constraints ensuring the global
nonnegativity of $S(k)$ is as follows. Define
\begin{equation}
 A=\rho v_dR^d+\sum_{j=1}^N z_jv_d(r_j^d-r_{j-1}^d).
 \label{eq:finite-band-mass-bound}
\end{equation}
Choose $0<\tau<1$, a spacing $\Delta>0$, and $K=L\Delta$ with
$L$ a positive integer. Maximize $\rho$ subject to nonnegativity of the
variables and
\begin{align}
 S(\ell\Delta)&\geq\tau+2\pi R\Delta A,
       &&\ell=0,\ldots,L-1,\label{eq:finite-band-grid}\\
 K^{-d/2}\left[\rho R^{d/2}
    +\sum_{j=1}^N z_j(r_j^{d/2}+r_{j-1}^{d/2})\right]
       &\leq1-\tau.
       \label{eq:finite-band-tail}
\end{align}
Indeed, differentiation under the Fourier integral gives
$|S'(k)|\leq2\pi RA$, so Eq.~\eqref{eq:finite-band-grid} ensures
$S(k)\geq\tau$ between successive grid points. The bound
$|J_{d/2}(x)|\leq1$ gives the same conclusion for $k\geq K$ from
Eq.~\eqref{eq:finite-band-tail}. These constraints therefore control all
wavenumbers, including those beyond the finite grid.

One searches over rational partitions and choices of $\tau,\Delta,K$,
using conservative rational bounds for the Fourier coefficients in the
constraints. For any supplied target density below $p_d$, this search
eventually produces a finite table exceeding that target. The proof of
Proposition~\ref{prop:finite-band-optimum} starts from an optimal
generalized pair measure, whose existence follows from
Ref.~\cite[Proposition~3.6]{Co22}. A cutoff and smoothing procedure
produces an ordinary function with a flat tail and arbitrarily small
density loss. A strictly positive margin in $S(k)$ then permits
approximation by radial bands: the Fourier error is at most
$\rho\|g_2-g_{2,\mathrm{approx}}\|_{L^1}$, uniformly in $k$.
The band parameters are obtained algorithmically; they are not given by
a closed formula in $d$.

\subsection{Decorrelation and the Range of the Pair Correlations}

The family in Eq.~\eqref{eq:finite-band-pair} satisfies the ordinary
pair decorrelation property in the exact form $g_2(r)=1$ for $r\geq R$.
An integrated, or generalized, decorrelation statement follows directly.
Writing $Z(s)$ for the cumulative coordination specified by $g_2$, the
excess cumulative coordination is
\begin{equation}
 M_d(s)=Z(s)-\rho v_ds^d
   =\rho\int_{|\bm r|\leq s}[g_2(r)-1]\,\rmd\bm r
   =S(0)-1,\qquad s\geq R.
 \label{eq:finite-band-excess}
\end{equation}
Thus no further excess coordination accumulates beyond $R$. The constant
equals $-1$ when $S(0)=0$; ordinary decorrelation alone does not impose
this additional hyperuniformity condition. Here the integrated statement
is a consequence of the flat tail, and is not an extra assumption on the
construction. The finite-radius property in each dimension also does not
by itself imply pointwise convergence $g_{2,d}(r)\to1$ at every fixed
$r>1$ as $d\to\infty$.

There is a stronger quantitative statement about the possible radius
$R$. Let $\phi_{\mathrm{band}}(d;R)$ be the supremum of
$\rho v_d/2^d$ over functions of the form
\eqref{eq:finite-band-pair} with $S(k)\geq0$ and $g_2(r)=1$ for
$r\geq R$. Then
\begin{equation}
 \left(1-\frac2R\right)^d\phi_{\mathrm{CE}}(d)
 \leq\phi_{\mathrm{band}}(d;R)
 \leq\phi_{\mathrm{CE}}(d),\qquad R>2.
 \label{eq:finite-band-range}
\end{equation}
The lower bound follows from a transformation based on volumes of
intersections of balls. It preserves the hard core and both
nonnegativity conditions, makes $g_2(r)$ identically one for $r\geq R$,
and multiplies the density by $(1-2/R)^d$. Approximating the transformed
function by radial bands gives the stated inequality for the supremum.
Consequently, any prescribed sequence $R_d>2$ tending to infinity,
however slowly, is compatible with the optimal exponent in
Eq.~\eqref{eq:finite-band-density-rate}: the additional contribution
$-\log_2(1-2/R_d)$ to the exponent tends to zero. For example, one may
take $R_d=4+2\log\log(d+e)$. No claim of an optimal exponent with a
dimension-independent cutoff follows from this estimate.

Detailed analytic proofs of the constructions in
Secs.~\ref{sec:radau-construction} and \ref{sec:finite-band-construction},
including Eq.~\eqref{eq:finite-band-range}, will be given in follow-up work.

\section{Discussion and conclusions}
\label{sec:discuss}

We have obtained a family of conjectural lower bounds on the maximal packing fraction of the form $\phi_{\max}\gtrsim e^{-c}d^{\beta}2^{-d}$ for every fixed $\beta>1$ and $c>1$ in sufficiently high dimensions. These bounds arise from a single family of hyperuniform pair correlation functions, Eq.~\eqref{eq:Zsigmaphi}, in which the contact weight $Z=d^{\beta}$ determines the polynomial enhancement while the gap $\sigma-1=c/d$ preserves it up to a dimension-independent factor. Thus, the different polynomial improvements do not require unrelated choices of $g_2(r)$, but are generated by varying a single parameter within one construction.

Our analysis shows that the structure-factor constraint requires $\beta<\beta_c(d)$, where $\beta_c(d)$ diverges with dimension. Consequently, every fixed exponent $\beta > 1$ eventually satisfies the two-point conditions. If $\beta$ is instead allowed to increase with $d$ toward $\beta_c(d)$, the resulting densities recover the previously conjectured Torquato--Stillinger exponential asymptotic form, including its subexponential correction up to a multiplicative constant.

The Gauss--Radau construction of Sec.~\ref{sec:radau-construction} shows that the Torquato--Stillinger exponential rate can be surpassed once the number of independently adjustable radial shells is allowed to grow with dimension.  The construction has $m=\lfloor(d-1)/4\rfloor$ positive delta-function shells, satisfies the ordinary decorrelation property $g_2(r)=1$ beyond a finite radius, and has terminal-density exponent $0.622556248918\ldots$.  This exponent lies strictly between the Torquato--Stillinger exponent $0.7786524795\ldots$ and the unrestricted generalized two-point optimum $0.6044005\ldots$.

Theorem~\ref{theorem:bound} establishes that every admissible radial Cohn--Elkies certificate has objective value at least as large as the Torquato--Stillinger exponential scale. This is a dual obstruction, not an evaluation of the unrestricted Cohn--Elkies optimum: it shows that the radial program cannot exclude candidate densities with exponent $\alpha_{\mathrm{TS}}=0.7786524795\ldots$. The analysis of Sec.~\ref{sec:universality} further shows why this exponent persists among the fixed-complexity, near-contact pair-correlation families considered here. By contrast, strong duality and the recent determination of the unrestricted generalized pair LP give the smaller exponent $\alpha_*=0.6044005\ldots$ in Eq.~\eqref{eq:finite-band-ce-exponent}. Thus $\alpha_{\mathrm{TS}}$ describes the universality class of the restricted ansatzes studied in this paper, not the optimum of the full two-point relaxation.

The finite-band construction of Sec.~\ref{sec:finite-band-construction} reaches the optimal exponential rate $\alpha_*=0.6044005\ldots$ using ordinary nonnegative pair correlation functions with a flat tail. In every fixed dimension, their admissible density parameters approach the Cohn--Elkies value arbitrarily closely. Moreover, Eq.~\eqref{eq:finite-band-range} shows that any prescribed cutoff radius tending to infinity with dimension suffices to retain the optimal exponent. Thus ordinary pair decorrelation in the form $g_2(r)=1$ beyond a finite radius places no restriction on the attainable two-point exponential rate. The corresponding integrated decorrelation statement follows from this flat tail, without imposing a separate generalized decorrelation hypothesis.

In summary, we have shown that there are test pair-correlation functions
$g_{2}$ that satisfy the hard-core condition,
$g_{2}(r)\geq 0$, $S(k)\geq 0$, and the decorrelation
principle, and that attains the optimal Cohn--Elkies pair-LP bound with no
primal--dual gap in the high-$d$ limit. However, this  establishes optimality only at the
\emph{two-point level} of the sphere-packing problem. This would not, by
itself, prove that $g_{2}$ is realizable by a translationally
invariant disordered packing. Realizability requires the existence of a
compatible hierarchy of $n$-particle correlation functions
$g_n$, $n\geq 3$, satisfying the appropriate symmetry,
consistency, hard-core, and positivity conditions. To promote the pair-LP
optimizer to an actual packing, one would need to show that, for every
fixed $n\geq 3$,
\begin{equation}
g_n
=
\mathcal{F}_n\!\left[\rho,g_2\right]
+
\varepsilon_{n},
\label{g2}
\end{equation}
where the functional $\mathcal{F}_n$ is determined solely by the density and pair
correlation, and the error $\varepsilon_{n}$ tends to zero in a norm
strong enough to preserve all relevant realizability inequalities. Without such a result, higher-order
constraints could still rule out the candidate $g_{2}$, even
though it is pair-LP feasible, pointwise decorrelating, and asymptotically
saturates the optimal upper bound.
Importantly,  if relation (\ref{g2}) were established, then it would imply that the densest sphere packings in the high-$d$ limit are disordered.
A natural next step is to seek sphere-packing constructions guided by these pair correlations.

\section*{Acknowledgement}
We are grateful to Henry Cohn for many illuminating discussions and collaborations on closely related topics and, in particular, for suggesting finding an analogue of Samorodnitsky’s proof, helping streamline it, and for suggesting the compact-range regularization of generalized pair measures underlying the finite-band construction.
The work of ST and CV is supported by the Army Research Office under Cooperative Agreement No. W911NF-22-2-0103.\\
We acknowledge the use of OpenAI's GPT-5.6 Sol to assist with proofreading the manuscript, identifying passages and arguments that could be clarified or improved, helping to develop the code used to generate figures, assisting the development of the range estimate in Sec.~\ref{sec:finite-band-construction}, and assisting in finding the Gauss--Radau construction in Sec.~\ref{sec:radau-construction}. All scientific arguments, results, and conclusions were independently reviewed and verified by the authors, who take full responsibility for the content of the manuscript.

\bibliographystyle{unsrt}
\bibliography{new.bib}

\end{document}